\documentclass[letterpaper,10pt,conference]{ieeeconf}

\usepackage{cite}
\usepackage{graphicx}
\usepackage{tikz}
\usepackage{subcaption}
\usepackage[cmex10]{amsmath}

\usepackage{algpseudocode}
\usepackage{algorithmicx} 
\usepackage{array}
\usepackage{multirow}

\usepackage{epsfig}
\usepackage{amsfonts,amssymb,multirow,bigstrut,booktabs,ctable,latexsym}

\usepackage[mathscr]{eucal}
\usepackage{verbatim}

\usepackage{amsthm}
\usepackage{algorithm}

\usepackage{stmaryrd}

\IEEEoverridecommandlockouts  
\begin{document}

\newtheorem{definition}{Definition}
\newtheorem{lemma}{Lemma}
\newtheorem{corollary}{Corollary}
\newtheorem{theorem}{Theorem}
\newtheorem{example}{Example}
\newtheorem{proposition}{Proposition}
\newtheorem{remark}{Remark}
\newtheorem{assumption}{Assumption}
\newtheorem{corrolary}{Corrolary}
\newtheorem{property}{Property}
\newtheorem{ex}{EX}
\newtheorem{problem}{Problem}
\newcommand{\argmin}{\arg\!\min}
\newcommand{\argmax}{\arg\!\max}
\newcommand{\st}{\text{s.t.}}
\newcommand \dd[1]  { \,\textrm d{#1}  }

\makeatother

\title{\Large\bf A Compositional Approach to Safety-Critical Resilient Control for Systems with Coupled Dynamics}


\author{Abdullah Al Maruf$^{1*}$, Luyao Niu$^{2*}$, Andrew Clark$^2$, J. Sukarno Mertoguno$^3$, and Radha Poovendran$^1$ %
\thanks{*Authors contributed equally to this work.}
\thanks{$^1$Abdullah Al Maruf and Radha Poovendran are with the Network Security Lab, Department of Electrical and Computer Engineering,
University of Washington, Seattle, WA 98195-2500
        {\tt\small \{maruf3e,rp3\}@uw.edu}}%
\thanks{$^2$Luyao Niu and Andrew Clark are with the Department of Electrical and Computer Engineering, Worcester Polytechnic Institute, Worcester, MA 01609
{\tt\small \{lniu,aclark\}@wpi.edu}}
\thanks{$^3$J. Sukarno Mertoguno is with Information and Cyber Sciences Directorate, Georgia Tech Research Institute, Atlanta, GA 30332
{\tt\small \{karno\}@gatech.edu}}}

\thispagestyle{empty}
\pagestyle{empty}

\maketitle

\begin{abstract}

Complex, interconnected Cyber-physical Systems (CPS) are increasingly common in applications including smart grids and transportation. Ensuring safety of interconnected systems whose dynamics are coupled is challenging because the effects of faults and attacks in one sub-system can propagate to other sub-systems and lead to safety violations. In this paper, we study the problem of safety-critical control for CPS with coupled dynamics when some sub-systems are subject to failure or attack. We first propose resilient-safety indices (RSIs) for the faulty or compromised sub-systems that bound the worst-case impacts of faulty or compromised sub-systems on a set of specified safety constraints. By incorporating the RSIs, we provide a sufficient condition for the synthesis of control policies in each failure- and attack- free sub-systems. The synthesized control policies compensate for the impacts of the faulty or compromised sub-systems to guarantee safety. We formulate sum-of-square optimization programs to compute the RSIs and the safety-ensuring control policies. We present a case study that applies our proposed approach on the temperature regulation of three coupled rooms. The case study demonstrates that control policies obtained using our algorithm guarantee system's safety constraints.

\end{abstract}
\section{Introduction}\label{sec:intro}

Safety is an important property of cyber-physical systems (CPS) in multiple domains including power systems and transportation \cite{ames2016control,cohen2020approximate,fan2020fast}. Safety violations can potentially cause damage to the system and even endanger human lives \cite{sullivan2017cyber,Jeep}. To this end, safety verification and safety-critical control for CPS have been extensively studied \cite{ames2016control,cohen2020approximate,xu2018constrained,prajna2007framework}.

CPS have been shown to be vulnerable to random failures and cyber attacks, which can cause safety violations \cite{sullivan2017cyber,Jeep}. To mitigate the impacts of faults and cyber attacks, defending mechanisms and resilient control for CPS have garnered significant research attention \cite{bjorck2015cyber,zhu2015game,ivanov2016attack,fawzi2014secure}. Resilient and safety-critical control can be even more challenging for the class of CPS that are formed by the interconnection of sub-systems \cite{rinaldi2001identifying,zhang2019robustness}. 
Due to the couplings among the sub-systems, faults and attacks in one sub-system may lead to safety violation in other sub-systems and the overall CPS. For instance, the blackout in India in 2012 was caused by the escalation of a local and small disturbance thorough the interconnections of the power system \cite{india}. Thus, an intelligent adversary can utilize such cascading effects to compromise the controllers and cause maximum-impact safety violations. In addition, the dimension of coupled sub-systems grows with the number of sub-systems, posing a scalability challenge in safety verification and safety-critical control design. 

To alleviate the scalability challenge, compositional approaches which decompose the safety constraint over the sub-systems have been proposed \cite{nejati2020compositional,sloth2012compositional,lyu2022small,coogan2014dissipativity,yang2015fault,yang2020exponential}. These approaches do not consider the presence of attack. Specifically, compositional approaches to fault-tolerant safety-critical control for coupled CPS has received limited research attention.



In this paper, we aim to develop a compositional approach to safety-critical control for interconnected systems in which some of the sub-systems are faulty or compromised. We consider a class of interconnected systems with the sub-systems' dynamics being coupled, which we will refer as interconnected system or coupled system interchangeably. We propose two types of resilient-safety indices (RSIs), named as {\it intrinsic resilient-safety index} (IRSI) and {\it coupled resilient-safety index} (CRSI), using the self- and coupled-dynamics of each sub-system, respectively. The sign and magnitude of RSIs characterize and quantify the impacts of faulty or compromised sub-systems on specified safety constraints. Using the proposed RSIs as well as control barrier functions, we provide conditions and algorithm to design control laws for the remaining sub-systems that are fault- and attack-free. We make the following contributions:


\begin{itemize}

\item We define RSIs for the sub-systems that are vulnerable to failure or attack. The IRSI bounds the worst-case impact of the sub-system on safety constraint due to intrinsic/self-dynamics, whereas the CRSI bounds the worst-case impact introduced by the couplings. 

\item Utilizing the proposed RSIs, we derive the control policies in the fault- and attack-free sub-systems. We prove that our proposed control policies guarantee safety in the presence of faulty or compromised sub-systems.

\item We propose an algorithm based on sum-of-squares optimization to compute the RSIs and the safety-ensuring control policy in each fault- and attack-free sub-system independently. We discuss two special cases of linear systems and monotone systems, where computationally efficient methods can be used to estimate the RSIs.

\item We present a case study of temperature regulation of interconnected rooms to illustrate our approach. We demonstrate that the control policy obtained using our algorithm maintains the specified safety constraints. 
\end{itemize}

The rest of the paper is organized as follows. Section \ref{sec:related} presents the related works. Section \ref{sec:formulation} presents the system model and formulates the problem. Section \ref{sec:Res} introduces the indices and presents the condition for safety-ensuring control policies. Section \ref{sec:algo} presents the algorithms for computing the indices and synthesizing the control policies. Section \ref{sec:simulation} contains a case study. Section \ref{sec:conclusion} concludes the paper.

\section{Related Work}\label{sec:related}

Safety-critical CPS are widely seen in real-world applications including power systems \cite{sullivan2017cyber}, robotics \cite{alemzadeh2016targeted}, and intelligent transportation \cite{Jeep,koscher2010experimental}. In the absence of fault or malicious adversary in the system, safety verification and synthesis have been studied using model checking \cite{clarke1997model} and deductive verification \cite{manna2012temporal}. Recently, barrier function based approaches, which map the safety constraint to a linear constraint on the control policy, have attracted extensive research attention \cite{ames2016control,xu2018constrained,prajna2007framework}.

The existing safety verification approaches focusing on the overall CPS state space become computationally demanding and even intractable \cite{mitchell2005time,prajna2007framework} when applied to coupled systems. Compositional approaches, which decompose the safety constraint to those defined over low-dimensional sub-systems, have been proposed \cite{nejati2020compositional,sloth2012compositional,lyu2022small,coogan2014dissipativity}. The formulations in \cite{nejati2020compositional,sloth2012compositional,lyu2022small,coogan2014dissipativity} do not consider the presence of random failures or malicious attacks.



To address the presence of adversary and random failures in CPS, resilient and fault-tolerant CPS have been studied. Typical approaches include employing defense mechanisms against malicious attacks \cite{pajic2014robustness,fawzi2014secure,cardenas2008research} and designing intrusion tolerant system architectures \cite{castro2002practical,verissimo2003intrusion,mertoguno2019physics,abdi2018guaranteed,niu2022verifying}. For complex interconnected CPS, the adversary can compromise the entire system by intruding into a subset of sub-systems and leveraging the interconnection.

For interconnected CPS under malicious attack, attack detectability is investigated in \cite{gallo2020detectability} for linear systems with pair-wise interconnections. How to guarantee safety under malicious attacks for interconnected CPS is not considered in \cite{gallo2020detectability}. In \cite{yang2020exponential,yang2015fault}, safety is achieved for interconnected systems by reconfiguring the control law of each sub-system and the coupling topology among them, assuming that each sub-system is exponentially stable in the absence of coupling. 
However, reconfiguring the control law and coupling topology may not always be feasible for CPS under malicious attack.
In addition, verifying the reconfiguration over all possible interdependencies can be computationally expensive when the system is of large-scale. 
\section{System Model and Problem Formulation}\label{sec:formulation}

This section presents the system model and our problem formulation. Consider a system $\mathcal{S}$ consisting of a finite set of interconnected sub-systems, denoted as $\{\mathcal{S}_i\}_{i=1}^N$. Each sub-system $\mathcal{S}_i$ has 
\begin{align}
    \mathcal{S}_i:\dot{x}_i=&f_{i, slf}(x_i)+g_{i,slf}(x_i) u_i\nonumber\\
    &+f_{i,cpl}(x_i,x_{-i})+g_{i,cpl}(x_i,x_{-i})u_i
\end{align}
where $x_i\in\mathbb{R}^{n_i}$ is the state of sub-system $\mathcal{S}_i$, $x_{-i}\in\mathbb{R}^{n-n_i}$ is the state of the other sub-systems excluding $\mathcal{S}_i$, $u_i\in\mathbb{R}^{r_i}$ is the input to the sub-system $\mathcal{S}_i$, and $n=\sum_{i=1}^Nn_i$. Functions $f_{i, slf}:\mathbb{R}^{n_i}\rightarrow\mathbb{R}^{n_i}$, $g_{i,slf}:\mathbb{R}^{n_i}\rightarrow\mathbb{R}^{n_i\times r_i}$,  $f_{i,cpl}:\mathbb{R}^{n_i}\times\mathbb{R}^{n-n_i} \rightarrow\mathbb{R}^{n_i}$, and $g_{i,cpl}:\mathbb{R}^{n_i}\times \mathbb{R}^{n-n_i}\rightarrow\mathbb{R}^{n_i\times r_i}$ are Lipschitz. Note that functions $f_{i, slf}$ and $g_{i,slf}$ are dependent on the states of sub-system $\mathcal{S}_i$ only, and we refer to the term 
\begin{equation*}
    F_{i, slf}(x_i,u_i)\triangleq f_{i, slf}(x_i)+g_{i,slf}(x_i) u_i
\end{equation*}
as the \textit{self-dynamics} of sub-system $\mathcal{S}_i$. Since functions $f_{i,cpl}$ and $g_{i,cpl}$ are jointly determined by the states of $x_i$ and those of other sub-systems $x_{-i}$, we refer to 
\begin{equation*}
    F_{i, cpl}(x,u_i)\triangleq f_{i,cpl}(x_i,x_{-i})+g_{i,cpl}(x_i,x_{-i})u_i
\end{equation*}
as the \textit{coupled-dynamics} of sub-system $\mathcal{S}_i$. We further assume that the inputs to each sub-systems are bounded as $u_i \in \mathcal{U}_i$ where $\mathcal{U}_i=\prod_{j=1}^{r_i}[\underline{u}_{i,j},\overline{u}_{i,j}]$ with $\underline{u}_{i,j}<\overline{u}_{i,j}$. A control policy for sub-system $\mathcal{S}_i$ is a function $\mu_i:\mathbb{R}^n\rightarrow \mathcal{U}_i$ that maps from the set of system states to the set of control inputs.

Let $x=[x_1^\top ~\cdots~x_N^\top]^\top\in\mathbb{R}^n$ and $u=[u_1^\top~\cdots~u_N^\top]^\top\in\mathbb{R}^r$ where $r=\sum_{i=1}^Nr_i$. Then the dynamics of $\mathcal{S}$ can be written as
\begin{equation}\label{eq:joint dynamics}
    \mathcal{S}:\begin{bmatrix}\dot{x}_1\\
    \vdots\\
    \dot{x}_N
    \end{bmatrix}=\begin{bmatrix}
    F_1(x_1,x_{-1},u_1)\\
    \vdots\\
    F_N(x_N,x_{-N},u_N)
    \end{bmatrix}\triangleq F(x,u)
\end{equation}
where $F_i(x_i,x_{-i},u_i)=F_{i, slf}(x_i,u_i) +F_{i, cpl}(x,u_i)$.

We consider that system $\mathcal{S}$ is given $K$ safety constraints for all time $t\geq 0$ where $K \in \mathbb{Z}_{+}$. We suppose each safety constraint is represented as $h^k(x) \geq 0$ where $h^k:\mathbb{R}^n\rightarrow \mathbb{R}$ is a continuously differentiable function for each $k=1,\ldots,K$. We denote the corresponding safety set as $\mathcal{C}$ so that $\mathcal{C}=\cap_{k=1}^N \{x\in \mathbb{R}^n: h^k(x) \geq 0\}$. We assume $\mathcal{C}$ is compact.

Some sub-systems are subject to random failures and malicious attacks. In these scenarios, the actuator of a failed or attacked sub-system $\mathcal{S}_i$ does not behave as expected. In the following, we assume that there exists a subset of sub-systems such that they are protected and does not incur random failures or attack. We denote this set of sub-systems as $\{\mathcal{S}_i: i \in \mathcal{N}_1\}$ where $\mathcal{N}_1\subseteq\{1, 2, \ldots, N\}$ and refer to them as \textit{protected sub-systems}. In addition, the set of remaining sub-systems that are subject to random failures and malicious attack is denoted as $\{\mathcal{S}_i: i \in \mathcal{N}_2\}$ where $\mathcal{N}_2\subseteq\{1,\ldots,N\}$ and we to refer them as \textit{vulnerable sub-systems}. Note, $\mathcal{N}_1 \cup \mathcal{N}_2 = \{1, 2, \ldots, N\}$ and  $\mathcal{N}_1 \cap \mathcal{N}_2 = \emptyset$.

Each vulnerable sub-system $\mathcal{S}_i$ where $i\in\mathcal{N}_2$ may incur fault or cyber attack initiated by a malicious adversary. The fault or attack can alter the control input $u_i$ injected to $\mathcal{S}_i$ to arbitrary $\Tilde{u}_i\in\mathcal{U}_i$. Since the system is interconnected, the altered behaviors from sub-systems in $\mathcal{N}_2$ can further propagate to other protected sub-systems, leading to potential safety violation if the sub-systems in $\mathcal{N}_1$ are not properly controlled.

In this paper we aim at computing control policies for the protected sub-systems to ensure safety of system $\mathcal{S}$ irrespective of the states and inputs of the vulnerable sub-systems. We state the problem as follows:
\begin{problem} \label{prob1}
Consider an interconnected system $\mathcal{S}$ where the sub-systems in $\{\mathcal{S}_i: i \in \mathcal{N}_2\}$ are subject to failures and malicious attacks while the sub-systems in $\{\mathcal{S}_i: i \in \mathcal{N}_1\}$ are protected. Synthesize a control policy  $\mu_i:\mathbb{R}^n\rightarrow\mathcal{U}_i$ for each $i \in \mathcal{N}_1$ such that system $\mathcal{S}$ is safe with respect to $\mathcal{C}$. 
\end{problem}


\section{RSIs and RSI-based Safety Guarantee} \label{sec:Res}

In this section, we first propose resilient-safety indices (RSIs) which relate to the worst-case impacts on the safety constraints caused by the self-dynamics and the coupled-dynamics of the vulnerable sub-systems. Based on the RSI we then obtain the sufficient condition for control policies in the protected sub-systems so that safety constraints are guaranteed irrespective of the condition (i.e. whether being faulty/compromised or not) of any of the vulnerable sub-systems.

We first define the RSIs related to the self-dynamics of vulnerable sub-systems $\{\mathcal{S}_i: i\in \mathcal{N}_2 \}$ as follows:
\begin{definition}\label{def:gamma}
For each $i \in \mathcal{N}_2$ and $k=1,\ldots,K$, we define an intrinsic resilient-safety index (IRSI) of sub-system $\mathcal{S}_i$ with respect to function $h^k$ as
\begin{equation}\label{eq:gamma1}
\hat{\gamma}_i^k = \inf_{x\in\mathcal{C}, u_i\in\mathcal{U}_i} \left\{\frac{\partial h^k}{\partial x_i}F_{i,slf}(x_i,u_i)\right\} 
\end{equation}
\end{definition}
The IRSI $\hat{\gamma}_i^k$ models the worst-case impact from the self-dynamics of sub-system $\mathcal{S}_i$ on the safety constraint $h^k(x) \geq 0$. 
The non-negative value of $\hat{\gamma}_i^k$ indicates that the sub-system $\mathcal{S}_i$ is intrinsically resilient-safe with respect to constraint $h^k(x) \geq 0$ as the self-dynamics do not contribute to the safety violation of $h^k(x) \geq 0$ for any $u_i\in\mathcal{U}_i$. The negative value of $\hat{\gamma}_i^k$ indicates that the self-dynamics of $\mathcal{S}_i$ can potentially cause violation to the safety constraint $h^k(x) \geq 0$ in the presence of attack or fault (the smaller $\hat{\gamma}_i^k$ is, the more detrimental $\mathcal{S}_i$ intrinsically is in violating $h^k(x) \geq 0$).

However, when $\hat{\gamma}_i^k$ is not available or not easy to compute, we may instead approximate it by finding a bound $\gamma_i^k\in\mathbb{R}$ such that for all $x\in\mathcal{C}$ and  $u_i\in\mathcal{U}_i$
\begin{equation}\label{eq:gamma}
    \frac{\partial h^k}{\partial x_i}F_{i,slf}(x_i,u_i) \geq \gamma_i^k.
\end{equation}
By Definition \ref{def:gamma}, we have $\hat{\gamma}_i^k \geq \gamma_i^k$ for any $\gamma_i^k$ satisfying \eqref{eq:gamma}. 

Now we define the RSIs related to the coupled-dynamics of vulnerable sub-systems $\{\mathcal{S}_i: i\in \mathcal{N}_2 \}$ as follows:
\begin{definition}\label{def:beta}
For each $k=1,\ldots,K$, we define a coupled resilient-safety index (CRSI) for all vulnerable sub-systems $\mathcal{S}_i$ with respect to function $h^k$ as
\begin{align}\label{eq:beta1}
\hat{\beta}^k=\inf_{x\in\mathcal{C},u_i\in\mathcal{U}_i} \left\{\sum_{i\in\mathcal{N}_2}\frac{\partial h}{\partial x_i}F_{i,cpl}(x,u_i)\right\}
\end{align}
\end{definition}
The CRSI $\hat{\beta}^k$ models the worst-case impact from coupled dynamics of vulnerable sub-systems on the safety constraint $h^k(x) \geq 0$. 
The non-negative value of $\hat{\beta}^k$ indicates that the coupled-dynamics of the vulnerable sub-systems do not contribute to the safety violation of $h^k(x) \geq 0$ for any $u_i\in\mathcal{U}_i$ where $i \in \mathcal{N}_2$. The negative value of $\hat{\beta}^k$ indicates that the coupled-dynamics of the vulnerable sub-systems can potentially cause violation to the safety constraint $h^k(x) \geq 0$ in presence of attack or fault (the smaller $\hat{\beta}^k$ is, the more detrimental the coupled-dynamics of the vulnerable sub-systems are in violating $h^k(x) \geq 0$).

However, when $\hat{\beta}_i^k$ is not available or not easy to compute, we may approximate it by finding $\beta^k\in\mathbb{R}$ such that for all $x\in\mathcal{C}$ and  $u_i\in\mathcal{U}_i$
\begin{equation}\label{eq:beta}
\sum_{i\in\mathcal{N}_2}\frac{\partial h^k}{\partial x_i}F_{i,cpl}(x,u_i) \geq \beta^k
\end{equation}
By Definition \ref{def:beta}, we have $\hat{\beta}^k \geq \beta^k$ for any $\beta^k$ satisfying \eqref{eq:beta}. 


Next, we present our main result based on the IRSI and CRSI given in Definition \ref{def:gamma} and \ref{def:beta}. We derive a sufficient condition for control policies in the protected sub-systems such that system $\mathcal{S}$ satisfies all the safety constraints. The result is formalized below. 

\begin{theorem} \label{Th1}
Suppose there exist constants $\alpha_i^k\in[0,1]$ and a control policy $\mu_i:\mathbb{R}^n\rightarrow\mathcal{U}_i$ for each $i\in \mathcal{N}_1$ such that the following holds for all $i \in \mathcal{N}_1$ and $k=1, \ldots, K$:
\begin{align}
     & \frac{\partial h^k}{\partial x_i}F_{i}(x_i,x_{-i},u_i)
       \geq \alpha_i^k\left(-\eta_i^k(h^k(x))-\beta^k-\sum_{j\in\mathcal{N}_2}\gamma_j^k\right)  \label{eq:MR}
\end{align}
where $\gamma_j^k$, $\beta^k$ are given in Eqn. \eqref{eq:gamma} and \eqref{eq:beta}, $\eta_i^k$ is an extended class $\mathcal{K}$ function and $\sum_{i\in\mathcal{N}_1} \alpha_i^k=1$ for each $k=1, \ldots, K$. Then the interconnected system $\mathcal{S}$ is safe with respect to $\mathcal{C}$ for all $t \geq 0$ by taking control policy $\mu_i$ at each $i \in \mathcal{N}_1$ given that $x(0) \in \mathcal{C}$.
\end{theorem}
\begin{proof}
According to \eqref{eq:joint dynamics} we can compute $\frac{\partial h^k}{\partial x}F(x,u)$ for each $k=1,\ldots,K$ as
\begin{align}
    &\frac{\partial h^k}{\partial x}F(x,u) = \sum_{i=1}^N\frac{\partial h^k}{\partial x_i}F_i(x_i,x_{-i},u_i) \nonumber\\
    =&\sum_{i\in\mathcal{N}_1}\frac{\partial h^k}{\partial x_i}F_i(x_i,x_{-i},u_i)+\sum_{i\in\mathcal{N}_2}\frac{\partial h^k}{\partial x_i}F_i(x_i,x_{-i},u_i)\nonumber\\
    \geq &\sum_{i\in\mathcal{N}_1}\alpha_i^k[ -\eta_i^k(h^k(x))-\beta^k-\sum_{j\in\mathcal{N}_2}\gamma_j^k]\nonumber\\
    &\quad+\sum_{i\in\mathcal{N}_2}\frac{\partial h^k}{\partial x_i}[F_{i,slf}(x_i,u_i)+F_{i,cpl}(x_i,u_i)]\nonumber\\
    =& -\sum_{i\in\mathcal{N}_1}\alpha_i^k \eta_i^k(h^k(x)) +\sum_{i\in\mathcal{N}_2}\Big(\frac{\partial h^k}{\partial x_i}F_{i,slf}(x_i,u_i) - \gamma_i^k\Big)\nonumber\\
    &\quad\quad+\Big(\sum_{i\in\mathcal{N}_2} \frac{\partial h^k}{\partial x_i}F_{i,cpl}(x_i,,u_i) - \beta^k\Big)\nonumber\\
    \geq &-\sum_{i\in\mathcal{N}_1}\alpha_i^k \eta_i^k(h^k(x)) \nonumber
\end{align}
where the last inequality holds by Eqn. \eqref{eq:gamma} and \eqref{eq:beta}. Since that $\eta_i^k$ is an extended class $\mathcal{K}$ function and $\alpha_i^k\in[0,1]$, we have that $\sum_{i\in\mathcal{N}_1}\alpha_i^k\eta_i^k$ is also an extended class $\mathcal{K}$ function. Thus $\frac{\partial h^k}{\partial x}F(x,u)\geq -\eta^k(h^k(x))$. Using the property of control barrier function \cite{ames2016control} and the assumption that $x(0) \in \mathcal{C}$, we have that the coupled system satisfies that $h^k(x)\geq 0$ for all $k=1, \dots, K$ and $t\geq 0$. Therefore we have that system $\mathcal{S}$ is safe with respect to $\mathcal{C}$.
\end{proof}

In Theorem $\ref{Th1}$, constant $\alpha_i^k$ specifies the weight on each protected sub-system $\mathcal{S}_i$ to satisfy the safety constraint  $h^k(x) \geq 0$. For example, $\alpha_i^k=1$ means that we let sub-system $\mathcal{S}_i$ be solely responsible for maintaining safety constraint $h^k(x) \geq 0$ among all the protected sub-systems. Parameter $\alpha_i^k$ for $i=1, \ldots, m$ and $k=1, \ldots, K$ are chosen in a way so that the condition \eqref{eq:MR} is satisfied. 

Theorem 1 implies that if the approximated RSIs $\gamma_i^k$ and $\beta^k$ for the vulnerable sub-systems are known, then the control policies in the protected sub-systems that guarantee system's safety can be calculated without knowing the exact models of the vulnerable sub-systems. This is useful since the control input $\tilde{u}_i$ employed in a sub-system compromised by adversarial attack is usually unknown. However, from Theorem \ref{Th1}, it is clear that one should use true values of RSIs, i.e., $\hat{\gamma}_i^k$ and $\hat{\beta}^k$ as $\gamma_i^k$ and $\beta^k$ so that condition \eqref{eq:MR} is less restrictive. RSIs or their approximations can be computed either numerically or analytically. Below we present a simple example for which we find the closed-form expressions for $\gamma_i^k$ and $\beta^k$. This will help us to gain some insights on RSIs.

\medskip

\noindent \textbf{Example:} Consider a system $\mathcal{S}:\dot{x} = Ax+Bu$ where 
\begin{equation*}
    A=\begin{bmatrix}
    a_{11}&\ldots&a_{1N}\\
    \vdots&\ddots&\vdots\\
    a_{N1}&\ldots&a_{NN}
    \end{bmatrix},\quad B=\begin{bmatrix}
    b_{11}&\ldots&0\\
    \vdots&\ddots&\vdots\\
    0&\ldots&b_{NN}
    \end{bmatrix}
\end{equation*}
Here $A \in \mathbb{R}^{N \times N}$, $B \in \mathbb{R}^{N \times N}$, $x=[x_1 \hdots x_N]^\top\in\mathbb{R}^n$ and $u=[u_1 \hdots u_N]^\top$. Matrix $A$ represents the system matrix for a synchronization dynamics with $a_{ii} = -\sum_{j=1, j \neq i}^N a_{ij}<0$. We consider that the system is given one ellipsoid safety constraint $\mathcal{C}=\{x \in \mathbb{R}^N : h(x) \geq 0 \}$ where $h(x)= 1- \sum_{i=1}^N c_i x_i^2$ and $c_i > 0$. Let the constraints on input be $\mathcal{U}_i \in [-1,1]$. Let $\mathcal{N}_1=\{1, \ldots, m\}$ and $\mathcal{N}_2=\{m+1,\ldots, N\}$. Then, we can write 
\begin{align*}
&\frac{\partial h}{\partial x_i}F_{i,slf}(x_i,u_i)= -2 c_i x_i (a_{ii}x_i+b_{ii}u_i)\\
&\quad\quad=-2 c_i a_{ii} (x_i+\frac{b_{ii}u_i}{2a_{ii}})^2+\frac{b_{ii}^2c_iu_i^2}{2a_{ii}} \\
&\beta=\sum_{i=m+1}^N \big(-2c_ix_i\sum_{j=1,j\neq i}^N(a_{ij}x_j)\big)\\
&\quad\quad\geq \sum_{i=m+1}^N 2 c_i|a_{ii}|~ \min_{x \in \mathcal{C}}(x_ix_j)
\end{align*}
Then with some efforts it can be shown that $\gamma_i=-\frac{b_{ii}^2c_i}{2|a_{ii}|}$ and $\beta=- \frac{c_{\mathcal{N}_2}^{max}}{c^{min}} \sum_{i=m+1}^N|a_{ii}|$ satisfy the conditions \eqref{eq:gamma} and \eqref{eq:beta} respectively where $c_{\mathcal{N}_2}^{max}$ is the maximum value in $\{c_{m+1}, \ldots c_N\}$  and $c^{min}$ is the minimum value in $\{c_1, \ldots, c_N\}$. The expression of $\gamma_i$ implies that $\gamma_i$ is inversely related to the magnitude of the eigenvalue of the self-dynamics of that sub-system $\mathcal{S}_i$ (i.e. the higher the magnitude of eigenvalue or faster decreasing the dynamics, the lower $\gamma_i$). This indicates that a sub-system with fast decreasing self-dynamics plays less significant role in violating the safety constraint under fault or attack. Now, the expression of $\beta$ implies that the higher the coupling (i.e. $|a_{ii}| = |\sum_{j=1, j \neq i}^N a_{ij}|$ for $i \in \mathcal{N}_2$), the higher the magnitude of $\beta$ becomes. Therefore, the sub-systems that have higher couplings will be critical in violating the safety constraint under fault or attack.

\section{Algorithmic computation of RSIs and control policies} \label{sec:algo}

In this section we present algorithms to compute the control policies and associated RSIs i.e. $\gamma_i^k$ and $\beta_k$. Our proposed algorithms are based on sum-of-squares (SOS) optimization. Later, we present two special classes of systems for which RSIs can be computed more efficiently. 

In the remainder of this section, we make the following assumption on the system dynamics and the given safety constraints.

\begin{assumption}\label{assump:semi-algebraic}
We assume that $F(x,u)$ is polynomial in $x$ and $u$, and $h^k(x)$ is polynomial in $x$ for $k=1,\ldots,K$ respectively.
\end{assumption}

Based on the above assumption, we now aim to compute IRSI $\hat{\gamma}_i^k$ and CRSI $\hat{\beta}^k$ using SOS optimization. However, in general IRSI and CRSI are difficult to compute. Therefore we relax the condition and focus on finding approximated values of IRSI and CRSI. Specifically, we find $\gamma_i^k$ and $\beta^k$ so that inequalities \eqref{eq:gamma} and \eqref{eq:beta} are reasonably tight. For that we first show how to translate conditions \eqref{eq:gamma} and \eqref{eq:beta} into SOS constraints. Then we minimize $\gamma_i^k$ and $\beta^k$ in the SOS formulation to make \eqref{eq:gamma} and \eqref{eq:beta} reasonably tight. 

Following results formalize our computation method of $\gamma_i^k$ for each $i\in \mathcal{N}_2$ and $k \in \{1,\ldots, K\}$.

\begin{lemma} \label{lemma:1}
Suppose Assumption \ref{assump:semi-algebraic} holds and $p_s(x,u_i)$, $w_j(x,u_i)$ and $v_j(x,u_i)$ are SOS polynomials where $j=1, 2, \ldots , r_i$ and $s= 1, \ldots ,K$. For $i\in \mathcal{N}_2$ and $k \in \{1,\ldots,K\}$ if  $\gamma_i^k$ is the solution to the sum-of-squares program:
\begin{subequations}\label{eq:gamma SOS}
\begin{align}
    &\min_{\gamma_i^k}~ -\gamma_i^k\\
    &\st\frac{\partial h^k}{\partial x_i}F_{i,slf}(x_i,u_i)-\gamma_i^k-\sum_{s=1}^k p_s(x,u_i)h^s(x)\nonumber\\
    &\quad-\sum_{j=1}^{r_i}\big(w_j(x,u_i)(u_{i,j}-\underline{u}_{i,j})+v_j(x,u_i)(\overline{u}_{i,j}-u_{i,j})\big)\nonumber\\
    &\quad\quad\text{ is SOS} \label{eq:gamma SOS 1}
\end{align}
\end{subequations}
then $\gamma_i^k$ satisfies Eqn. \eqref{eq:gamma}. Furthermore $\gamma_i^k=\hat{\gamma}_i^k$ when expressions in \eqref{eq:gamma SOS 1} is quadratic.
\end{lemma}
\begin{proof}
Let $\gamma_i^k$ be the solution to SOS program \eqref{eq:gamma SOS}. Since $p_s(x,u_i)$, $w_j(x,u_i)$ and $v_j(x,u_i)$ are SOS polynomials, we have that $\sum_{s=1}^k p_s(x,u_i)h^s(x)\geq 0$ for all $x\in \mathcal{C}$ and $\sum_{j=1}^{r_i}\big(w_j(x,u_i)(u_{i,j}-\underline{u}_{i,j})+v_j(x,u_i)(\overline{u}_{i,j}-u_{i,j})) \geq 0$ for all $u_i\in\mathcal{U}_i$. Thus any $\gamma_i^k$ rendering constraint \eqref{eq:gamma SOS 1} an SOS satisfies that
\begin{equation*}
    \frac{\partial h^k}{\partial x_i}F_{i,slf}(x_i,u_i) \geq \gamma_i^k ,~\forall x\in \mathcal{C},u_i\in\mathcal{U}_i
\end{equation*}
Now suppose the case when expressions in \eqref{eq:gamma SOS 1} is quadratic. Assume that $\gamma_i^k \neq \hat{\gamma}_i^k$. Then $\hat{\gamma}_i^k>\gamma_i^k$ by the definition of $\hat{\gamma}_i^k$. Since for quadratic polynomial SOS and non-negativity is equivalent \cite{powers2011positive}, $\hat{\gamma}_i^k$ is also feasible to constraint \eqref{eq:gamma SOS 1}. However, this contradicts the optimality of $\gamma_i^k$ to SOS program \eqref{eq:gamma SOS}. Hence $\gamma_i^k = \hat{\gamma}_i^k$ in this case.
\end{proof}

Similarly we use the following result to compute $\beta^k$ for each $k\in\{1,\ldots,K\}$.
\begin{lemma}\label{lemma:beta SOS}
Suppose Assumption \ref{assump:semi-algebraic} holds and $p_s(x,u_i)$, $w_{i,j}(x,u)$ and $v_{i,j}(x,u)$ are SOS polynomials where $i\in \mathcal{N}_2$, $j=1, 2, \ldots , r_i$ and $s= 1, \ldots ,K$. For $k \in \{1,\ldots,K\}$ if  $\beta^k$ is the solution to the sum-of-squares program:
\begin{subequations}\label{eq:beta SOS}
\begin{align}
    &\min_{\beta^k}~ -\beta^k\\
    &\st\sum_{i\in\mathcal{N}_2}\frac{\partial h^k}{\partial x_i}F_{i,cpl}(x,u_i) -\beta^k -\sum_{s=1}^k p_s(x,u_i)h^s(x)\nonumber\\
    &\quad\quad-\sum_{i\in\mathcal{N}_2}\sum_{j=1}^{r_i}\big(w_{i,j}(x,u)(u_{i,j}-\underline{u}_{i,j}) \nonumber\\
    &\quad\quad+v_{i,j}(x,u)(\overline{u}_{i,j}-u_{i,j})\big)\text{ is SOS}\label{eq:beta SOS 1}
\end{align}
\end{subequations}
then $\beta^k$ satisfies Eqn. \eqref{eq:beta}. Furthermore $\beta^k=\hat{\beta}^k$ when expressions in \eqref{eq:beta SOS 1} is quadratic.
\end{lemma}
\begin{proof}
Let $\beta^k$ be the solution to SOS program \eqref{eq:beta SOS}. Since $p_s(x,u)$, $w_{i,j}(x,u)$ and $v_{i,j}(x,u)$ are SOS polynomials, we have that $\sum_{s=1}^k p_s(x,u_i)h^s(x)\geq 0$ for all $x\in\mathcal{C}$ and $\sum_{i\in\mathcal{N}_2}\sum_{j=1}^{r_i}\big(w_{i,j}(x,u)(u_{i,j}-\underline{u}_{i,j})+v_{i,j}(x,u)(\overline{u}_{i,j}-u_{i,j})\big)$ for all $u_i\in\mathcal{U}_i$ and $i\in \mathcal{N}_2$. Thus $\beta^k$ satisfies that
\begin{multline*}
    \sum_{i\in\mathcal{N}_2}\frac{\partial h^k}{\partial x_i}F_{i,cpl}(x,u_i) \geq \beta^k,~\forall x\in\mathcal{C},u_i\in\mathcal{U}_i, i \in \mathcal{N}_2
\end{multline*}
We note that $\beta^k=\hat{\beta}^k$ when expressions in \eqref{eq:beta SOS 1} is quadratic using arguments similar to the proof of Lemma \ref{lemma:1}.
\end{proof}

Now we present our algorithm for computing control policies given the approximated IRSI and CRSI, i.e. $\gamma_i^k$ and $\beta^k$. To do so, we translate the condition given in \eqref{eq:MR} as SOS constraint and formulate an SOS program to compute the control input $u_i$ for each $i\in\mathcal{N}_1$. The following lemma presents the result.

\begin{lemma}
Suppose $\gamma_i^k$ and $\beta^k$ are given for each $i\in\mathcal{N}_2$ and $k\in\{1,\ldots,K\}$. Suppose the following expressions are SOS for each $i\in\mathcal{N}_1$ and $k\in\{1,\ldots,K\}$.
\begin{subequations}\label{eq:feasibility SOS}
\begin{align}
    &\frac{\partial h^k}{\partial x_i}[f_{i,slf}(x_i)+g_{i,slf}(x_i)\tau_i(x) + f_{i,cpl}(x_i,x_{-i})\nonumber \\
    & \quad +g_{i,cpl}(x_i,x_{-i})\tau_i(x)] - \alpha_i^k\Big(-\eta_i^k(h^k(x))-\beta^k \nonumber \\
    & \quad \quad -\sum_{j\in\mathcal{N}_2}\gamma_j^k\Big) -\sum_{s=1}^K\lambda_s(x)h^s(x),  \label{eq:feasibility SOS 1} \\
    &\tau_{i,j}(x)-\underline{u}_{i,j},\quad \overline{u}_{i,j}-\tau_{i,j}(x),~\forall j=1,\ldots,r_i,  \label{eq:feasibility SOS 2}
\end{align}
\end{subequations}
where $\lambda_s(x)$  is an SOS polynomial and $\tau_{i,j}(x)$ is a polynomial in $x$ for $i\in\mathcal{N}_1$, $j=1,\ldots,r_i$ and $s=1, \ldots, K$. Then condition \eqref{eq:MR} is satisfied for $x\in\mathcal{C}$ when $u_i$ is chosen as $u_i= \tau_i(x) =[\tau_{i,1}(x) \cdots \tau_{i,r_i}(x)]^T$ for all $i \in \mathcal{N}_1$.
\end{lemma}
\begin{proof}
Since $\lambda_s(x)$ is an SOS polynomial, we have that $\lambda_s(x)h^s(x)\geq 0$ implying $\sum_{s=1}^K\lambda_s(x)h^s(x) \geq 0$ for all $x \in \mathcal{C}$. When Eqn. \eqref{eq:feasibility SOS} is an SOS, we thus have that 
\begin{multline}
    \frac{\partial h^k}{\partial x_i}F_{i,slf}(x_i,\tau_i(x)) + F_{i,cpl}(x,\tau_i(x)) \nonumber\\
    - \alpha_i^k\Big(-\eta_i^k(h^k(x))-\beta^k  -\sum_{j\in\mathcal{N}_2}\gamma_j^k\Big) \geq 0
\end{multline}
holds for all $x\in \mathcal{C}$. By choosing $u_i=\tau_i(x)$ for all $i \in \mathcal{N}_1$, we recover the condition in Eqn. \eqref{eq:MR}.
\end{proof}

  \begin{center}
  	\begin{algorithm}[!htp]
  		\caption{Solution algorithm for finding control policy}
  		\label{algo:C1}
  		\begin{algorithmic}[1]
  			\State \textbf{Input}: Dynamics $F$. Functions $\{h^k\}_{k=1}^K$. Constants $\{\alpha_i^k\},\{\underline{u}_{i,j}\},\{\overline{u}_{i,j}\}$.
  			\State \textbf{Output:} $\{\gamma_i^k\}_{i\in\mathcal{N}_2}$, $\beta^k$, and control inputs $\{u_i\}$.
            \State Solve the SOS programs in Eqn. \eqref{eq:gamma SOS} and \eqref{eq:beta SOS} to calculate $\gamma_i^k$ and $\beta^k$ for each for each $i\in\mathcal{N}_2$ and $k\in\{1,\ldots,K\}$.
            \State  \textbf{Initialization:} $Flag \leftarrow 1$.
            \For{$i\in \mathcal{N}_1$}
            \State Solve \eqref{eq:feasibility SOS} using $\gamma_i^k$ and $\beta^k$ for all $k=1, \ldots K$. \label{linecheck}
            \If{Line \ref{linecheck} is feasible}
            \State $u_i\leftarrow\tau_i$. 
            \Else 
            \State $Flag \leftarrow 0$.
            \State \textbf{break}
            \EndIf
            \EndFor
            \If{$Flag=1$}
            \State \textbf{return} $\{\gamma_i^k\}$, $\{\beta^k\}$ and $\{u_i\}$.
            \Else 
            \State Not feasible.
            \EndIf
  		\end{algorithmic}
  	\end{algorithm}
  \end{center}

Now we present Algorithm \ref{algo:C1} that uses the above results to find the control policies in the protected sub-systems. First the approximated RSIs $\gamma_i^k$ and $\beta^k$ are calculated using Lemma \ref{lemma:1} and \ref{lemma:beta SOS} for each vulnerable sub-system and safety constraint. Then control policy $\mu_i$ for each $i \in \mathcal{N}_1$ is calculated from line 5 to line 12.
If the SOS conditions in Eqn. \eqref{eq:feasibility SOS} is feasible for all $i \in \mathcal{N}_1$ and $k \in \{1, \ldots K\}$, the algorithm returns the control policies for all the protected sub-systems. We remark that Algorithm \ref{algo:C1} can be implemented in an offline manner and it can calculate the control policy for each protected sub-system in parallel.


It may be possible that there exists no control policy that guarantees safety of the system depending on the cardinality of $\mathcal{N}_1$ or the ranges of $\mathcal{U}$ and $\mathcal{C}$ or the model dynamics $F(x,u)$. In that case Algorithm \ref{algo:C1} becomes infeasible, i.e. fails to find feasible control policies that guarantee safety. However, it may happen that Algorithm \ref{algo:C1} is infeasible for $\mathcal{C}$ but feasible for some set $\tilde{\mathcal{C}} \subset \mathcal{C}$. In such scenario one need to modify the algorithm and search for the set $\tilde{\mathcal{C}}$ such that the algorithm becomes feasible. In that case the system $\mathcal{S}$ will maintain safety for all $t\geq 0$ if $x(0) \in \tilde{\mathcal{C}}$.

Here we make a remark on the case when any of the safety constraints is local to a sub-system. Suppose there exist $\tilde{k} \in \{1, \ldots, k\}$ and $\tilde{i} \in \{1, \ldots, N\}$ such that $h^{\tilde{k}}:\mathbb{R}^{n_{\tilde{i}}}\rightarrow\mathbb{R}$ is local to sub-system $\mathcal{S}_{\tilde{i}}$ which can be written as $h^{\tilde{k}}(x_{\tilde{i}})$. This implies $\frac{\partial h^{\tilde{k}}}{\partial x_j}$ is a vector with all entries being zero for $j\neq \tilde{i}$. If $\tilde{i} \in \mathcal{N}_1$, then we have $\sum_{j \in \mathcal{N}_2}\hat{\gamma}_j^{\tilde{k}} = \hat{\beta}^{\tilde{k}}= 0$. This implies $|\mathcal{N}_1|-1$ conditions in Theorem \ref{Th1} are trivially satisfied and can be omitted from the SOS program. However, if $\tilde{i} \in \mathcal{N}_2$, then our approach can be modified to obtain less conservative SOS formulation. In this case, instead of calculating $\gamma_{\tilde{i}}^{\tilde{k}}$ (note, $\hat{\gamma}_j^{\tilde{k}} = 0$ for $j\neq {\tilde{i}}$) and $\beta^{\tilde{k}}$, we can impose the following constraint directly in each of the SOS program: $\frac{\partial h^{\tilde{k}}}{\partial x_{\tilde{i}}}F_{\tilde{i}}(x_{\tilde{i}}, x_{-{\tilde{i}}},u_{\tilde{i}}) \geq -\eta_{\tilde{i}}^{\tilde{k}}( h^{\tilde{k}}(x_{\tilde{i}}))$ for all $x\in \mathcal{C}$ and $u_{\tilde{i}} \in \mathcal{U}_{\tilde{i}}$. This condition is less conservative than the constraint $\gamma_{\tilde{i}}^{\tilde{k}} + \beta^{\tilde{k}} \geq -\eta_{\tilde{i}}^{\tilde{k}}( h^{\tilde{k}}(x_{\tilde{i}}))$ obtained via Theorem 1. We note that $\frac{\partial h^{\tilde{k}}}{\partial x_{\tilde{i}}}F_{\tilde{i}}(x_{\tilde{i}}, x_{-{\tilde{i}}},u_{\tilde{i}}) \geq -\eta_{\tilde{i}}^{\tilde{k}}( h^{\tilde{k}}(x_{\tilde{i}}))$ introduces infeasibility if the impact of compromised $u_{\tilde{i}} \in \mathcal{U}_{\tilde{i}}$ on $h^{\tilde{k}}(x_{\tilde{i}})$ cannot be compensated by the system states. For an example see Section \ref{sec:simulation}.

Now we present two special cases where the IRSI $\hat{\gamma}_i^k$ and CRSI $\hat{\beta}^k$ can be obtained more efficiently relative to SOS approach.

\subsubsection{LTI System with Half-Plane Constraint}
Consider system $\mathcal{S}$ be given as
\begin{equation}
    \mathcal{S}:\dot{x} = Ax+Bu,
\end{equation}
where 
\begin{equation*}
    A=\begin{bmatrix}
    A_{11}&\ldots&A_{1N}\\
    \vdots&\ddots&\vdots\\
    A_{N1}&\ldots&A_{NN}
    \end{bmatrix},\quad B=\begin{bmatrix}
    B_{11}&\ldots&\mathbf{0}\\
    \vdots&\ddots&\vdots\\
    \mathbf{0}&\ldots&B_{NN}
    \end{bmatrix}
\end{equation*}
$A_{ij}\in\mathbb{R}^{n_i\times n_j}$, and $B_{ii}\in\mathbb{R}^{n_i\times r_i}$. We consider $h^k(x)=a_k^\top x$ where $a_k\in\mathbb{R}^{n}$ for all $k=1,\ldots,K$ and set $\mathcal{C}$ is compact. We show that IRSI $\hat{\gamma}_i^k$ can be computed using a linear program when $\mathcal{U}_i=\{u_i:w^\top u_i\geq 0\}$. For each $i\in\mathcal{N}_2$, if $\gamma_i^k=a_k^\top F_{i,slf}(x_i^*,u_i^*)$, where $(x_i^*,u_i^*)$ is the solution to the following linear program:
\begin{subequations}\label{eq:gamma LP}
\begin{align}
    \min_{x_i,u_i}~ &a_k^\top F_{i,slf}(x_i,u_i)\\
    \st ~&a_k^\top x\geq 0,~\forall k\label{eq:gamma LP 1}\\
    &w^\top u_i\geq 0
\end{align}
\end{subequations}
then $\gamma_i^k=\hat{\gamma}_i^k$ as given in Definition \ref{def:gamma}. 
Let $\mathcal{P}$ be the polyhedron induced by the constraints in Eqn. \eqref{eq:gamma LP}. We have that $\hat{\gamma}_i^k$ is attained at some vertex of $\mathcal{P}$. Similarly we can consider the computation of CRSI $\hat{\beta}^k$. For each sub-system $i\in\mathcal{N}_2$, if $\beta^k=\sum_{i\in\mathcal{N}_2}a_k^\top F_{i,cpl}(x^*,u_i^*)$, where $[x^*,u_i^*]$ is the solution to the following linear program:
\begin{subequations}\label{eq:beta LP}
\begin{align}
    \min_{x,u_i}~ &\sum_{i\in\mathcal{N}_2}a_k^\top F_{i,cpl}(x,u_i)\\
    \st ~&a_k^\top x\geq 0,~\forall k\label{eq:beta LP 1}\\
    &w^\top u_i\geq 0
\end{align}
\end{subequations}
then $\beta^k=\hat{\beta}^k$ as given in Definition \ref{def:beta}.

\subsubsection{Monotone System with Hyperrectangle Constraints}
In the following, we consider monotone systems under hyperrectangle constraints. Consider system \eqref{eq:joint dynamics} and a safety set $\mathcal{C}=[\underline{x},\overline{x}]$, where $\underline{x},\overline{x}\in\mathbb{R}^n$. We consider $\underline{x}\leq \overline{x}$ element-wise, and thus $\mathcal{C}$ is a hyperrectangle. 
\begin{definition}[Monotonicity \cite{coogan2017finite}]\label{def:monotone}
Function $F(x,u)$ is monotone with respect to $x$ and $u$ if
\begin{equation}
    x\leq x' \text{ and } u\leq u'\implies F(x,u)\leq F(x',u')
\end{equation}
where the order relation $\leq$ is compared element-wise. 
\end{definition}

Monotonicity of a function $F(x,u)$ can be shown by verifying the signs of $\frac{\partial F}{\partial x}$ and $\frac{\partial F}{\partial u}$ \cite{coogan2017finite}. Using Definition \ref{def:monotone}, Eqn. \eqref{eq:gamma} and \eqref{eq:beta}, we have the following result:
\begin{lemma}
Consider system \eqref{eq:joint dynamics} and a hyperrectangle safety set $\mathcal{C}=[\underline{x},\overline{x}]$. If functions $\tilde{F}_{i,slf}(x,u_i)=\frac{\partial h^k}{\partial x_i} F_{i,slf}(x_i,u_i)$ and $\tilde{F}_{i,cpl}(x,u_i)=\frac{\partial h^k}{\partial x_i} F_{i,cpl}(x,u_i)$ are monotone with respect to $x$ and $u$ as given in Definition \ref{def:monotone}, then  for $i \in \mathcal{N}_2$ and $k=1,\ldots,K$ we have 
\begin{align*}
    \hat{\gamma}_i^k=\tilde{F}_{i,slf}(\underline{x},\underline{u}_i), \\
    \hat{\beta}^k=\sum_{i\in\mathcal{N}_2}\tilde{F}_{i,cpl}(\underline{x},\underline{u}_i).
\end{align*}
\end{lemma}
\section{Case Study} \label{sec:simulation}

In this section, we illustrate our proposed approach using an example on the temperature regulation in a circular building consisting of $N$ rooms \cite{meyer2017compositional}. For each room $i=1,\ldots,N$, we denote its temperature as $x_i$ which follows dynamics given as
\begin{equation}
    \dot{x}_i=\frac{1}{\delta}(w(x_{i+1}+x_{i-1}-2x_i) + y(T_e-x_i)+z(T_h-x_i)u_i),
\end{equation}
where $x_{i+1}$ and $x_{i-1}$ are the temperatures of the neighboring rooms, $T_e$ is the outside temperature, and $T_h$ is the heater temperature. For rooms $i=1$ and $i=N$, we let $x_{0}=x_N$ and $x_{N+1}=x_1$. In this case study, we consider $N=3$, $T_e=-1^\circ C$, and $T_h=50^\circ C$. We additionally let $\mathcal{U}_1\in[0,0.6]$ and $\mathcal{U}_2,\mathcal{U}_3\in[-2,2]$. The coefficients $w$, $y$, and $z$ are chosen as $w=0.45$, $y=0.045$, and $z=0.09$, respectively. Parameter $\delta$ is set as $\delta=0.1$. We consider that the controller of room $1$ is compromised via an adversarial attack. Rooms $2$ and $3$ are the protected sub-systems. In the remainder of this section, we study two scenarios.

In the first scenario, we consider that the system is given one safety constraint $x\in\mathcal{C}$ for all time $t\geq 0$, where $\mathcal{C}=\{x:h(x)\geq0\}$ and
\begin{equation*}
    h(x)=\left(\frac{\sum_{i=1}^3x_i}{3}-15\right)\left(20-\frac{\sum_{i=1}^3x_i}{3}\right).
\end{equation*}

Using the SOS program in Eqn. \eqref{eq:gamma SOS} and \eqref{eq:beta SOS}, we compute the approximated IRSI and CRSI. We have that $\gamma_1=-22.05$ and $\beta_1=-2.636$. We then synthesize the control input by enforcing constraint \eqref{eq:feasibility SOS}. We show the average temperature $\frac{\sum_{i=1}^3x_i}{3}$ at each time step in Fig. \ref{fig:global}. We observe that $\frac{\sum_{i=1}^3x_i}{3}\in[15,20]$ for all time steps and thus $h(x)\geq 0$ for all time.

\begin{figure}
    \centering
    \includegraphics[scale=0.45]{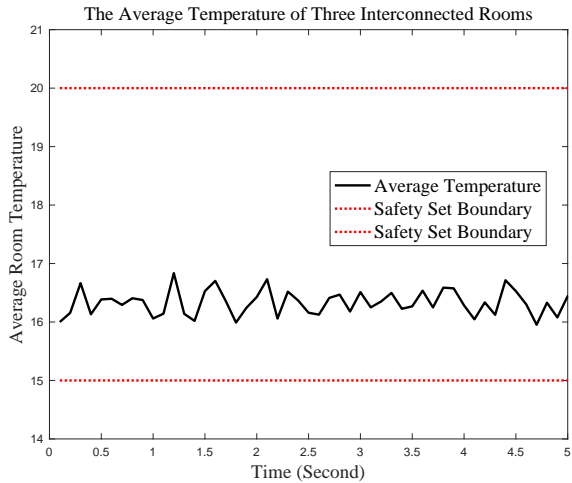}
    \caption{The average temperature of three rooms over $50$ time steps. Room $1$ is compromised and rooms $2,3$ are protected from failure and attack. The average temperature is depicted using solid black line. The boundaries of set $\mathcal{C}$ are shown using red dotted lines.}
    \label{fig:global}
\end{figure}

In the second scenario, we consider a safety set $\mathcal{C}=\{x:h^i(x_i)\geq 0, i=1,2,3\}$. Each function $h^i$ specifies a range for the temperature in room $i$. We consider 
\begin{align*}
    h^1(x_1) &= (16-x_1)(x_1-10),\\
    h^2(x_2) &= (22-x_2)(x_2-15),\\
    h^3(x_3) &= (25-x_3)(x_3-14).
\end{align*}
Since room $1$ is compromised, we can only satisfy the safety constraint $h_1$ by regulating the temperature of rooms $2$ and $3$ and utilizing the coupling term $w(x_{2}+x_{3}-2x_1)$. To ensure the satisfaction of $h^1(x_1)\geq 0$, we introduce the following constraint over $x_2$ and $x_3$ when synthesizing the control policies for rooms $2$ and $3$:
\begin{multline}\label{eq:translated constraint}
    -\frac{\partial h^1}{\partial x_1}\left[\frac{1}{\delta}(2wx_1 + yx_1-z(T_h-x_1)u_1)\right]-\eta^1(h^1(x_1))\\\leq\frac{\partial h^1}{\partial x_1}\left[\frac{1}{\delta}(w(x_{2}+x_{3}) + yT_e)\right] 
\end{multline}
where $\eta^1(\cdot)$ is a class $\mathcal{K}$ function. When $x_2$ and $x_3$ are chosen such that constraint \eqref{eq:translated constraint} is met regardless of $u_1$ for any $x_1\in[10,16]$, we have that $h^1(x_1)\geq 0$ is satisfied. To this end, we let $u_1$ be chosen as the worst-case one with different values of $x_1$. By imposing inequality \eqref{eq:translated constraint} as an additional constraint when synthesizing $u_2$ and $u_3$, we can guarantee the satisfaction of the safety constraint $x\in\mathcal{C}$, even though room $1$ is compromised. We remark that constraint \eqref{eq:translated constraint} needs to be compatible with constraints $h^2(x_2)\geq 0$ and $h^3(x_3)\geq 0$ to guarantee the feasibilities of $u_2$ and $u_3$. One can verify that when $u_1\geq 6.2498$, incorporating constraint \eqref{eq:translated constraint} leads to infeasibility when synthesizing $u_2$ and $u_3$.


Now, for the simulation we let the temperature in each room $i$ be at the boundary of their corresponding safety constraint at the first time step. We then compute the inputs $u_2$ and $u_3$ at each time step and depict the evolution of the temperature in each room in Fig. \ref{fig:temperature}. We plot the temperature of room $1$, $2$, and $3$ using black solid line, blue dash line, and red dash-dotted line, respectively. We observe that the temperature in each room $i$ always satisfies that $h^i(x_i)\geq 0$, indicating that the safety constraint is never violated even though the controller of room $1$ is compromised.
\begin{figure}
    \centering
    \includegraphics[scale = 0.45]{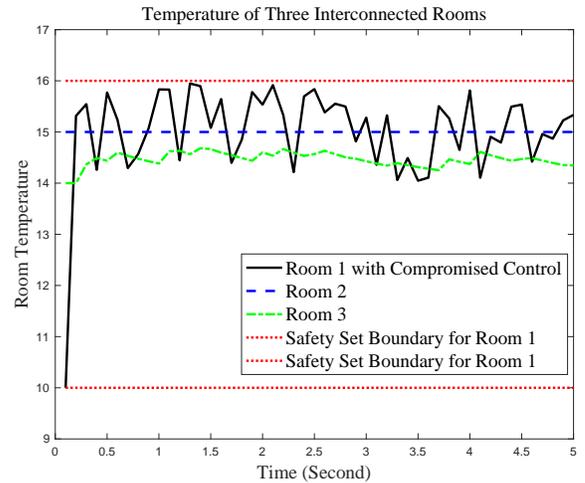}
    \caption{The temperature of each room $1$, $2$ and $3$ over $50$ time steps. Room $1$ is compromised and rooms $2,3$ are protected from failure and attack. The temperature of room $1$ is depicted using solid black line. The temperature of room $2$ is shown in blue dash line. The temperature of room $3$ is shown in red dash-dotted line.}
    \label{fig:temperature}
\end{figure}
\section{Conclusion}\label{sec:conclusion}

In this paper, we studied the problem of safety-critical control synthesis of CPS with multiple interconnected sub-systems. We considered that a set of sub-systems are vulnerable in the sense that their controllers may incur random failures or malicious attacks. For the vulnerable sub-systems we introduced resilient-safety indices (RSIs) bounding the worst-case impacts of vulnerable systems towards the specified safety constraints. The sign of RSI indicates the contribution of vulnerable sub-system in either satisfying or violating the corresponding safety constraint whereas the magnitude quantifies such contribution. We provided a sufficient condition for the control policies in the non-vulnerable sub-systems so that the safety constraints are satisfied in the presence of failure or attack in the vulnerable sub-systems. We formulated sum-of-squares optimization programs to compute the RSIs and safety-ensuring control policies. Control policy in each sub-system can be computed independently using our proposed algorithm. We presented two special cases for which the RSIs can be found more efficiently. We demonstrated the usefulness of our proposed approach using an example on temperature regulation of interconnected rooms.

\bibliographystyle{IEEEtran}
\bibliography{MyBib}

\end{document}